\documentclass[11pt]{amsart}
\usepackage{setspace} 
\usepackage{graphicx}

\usepackage{amssymb}

\usepackage{fancybox}
\usepackage{epstopdf}
\usepackage{pgf,tikz}
\usepackage[latin1]{inputenc}

\usepackage[figuresright]{rotating}
\usepackage{color}
\usepackage{amsmath, amsthm, amssymb, amsfonts}
\usepackage{geometry}
\usepackage{fleqn}
\usepackage{caption}
\newtheorem{Theorem}{Theorem}
\newtheorem{Example}{Example}[section]
\newtheorem{Remark}{Remark}

\newtheorem{Definition}{Definition}[section]
\newtheorem{Lemma}{Lemma}

\newtheorem{Proposition}{Proposition}
\newtheorem{Corollary}{Corollary}

\begin{document}

\title[On polycyclic codes over a finite chain ring]
{On polycyclic codes over a finite chain ring}%
\author[A. Fotue Tabue et al.]{A. Fotue Tabue\textsuperscript{1}, E. Martínez-Moro\textsuperscript{2}, T. Blackford\textsuperscript{3}}%

\address{\textsuperscript{1} Alexandre Fotue Tabue, Department of Mathematics,
Faculty of Science, University of Yaoundé 1, Yaoundé, Cameroon}
\email{alexfotue@gmail.com}

\address{\textsuperscript{2} Edgar Martínez-Moro, Institute of Mathematics, University of Valladolid, Spain}
\email{edgar.martinez@uva.es}

\address{\textsuperscript{3}Thomas Blackford, Department of Mathematics, Western Illinois University, 1 University Circle, Macomb, IL61455, USA}
\email{JT-Blackford@wiu.edu}

\keywords{Galois extension of finite chain rings; Polycyclic code;
Trace code; Restriction code. \\
{Edgar Martínez-Moro has been partially supported by the Spanish MINECO under Grant MTM2015-65764-C3-1}}

\begin{abstract} Galois images of polycyclic codes
over a finite chain ring $S$  and their annihilator dual are
investigated. The case when a polycyclic codes is  Galois-disjoint  over the ring
$S,$  is  characterized and, the trace codes and restrictions
of free polycyclic codes over $S$ are also determined givind  an analogue of Delsarte theorem among trace map, any S -linear code and
its annihilator dual.
\end{abstract}

\maketitle \emph{AMS Subject Classification 2010:} 13B02, 94B05.

\noindent


\section{Introduction}

 In \cite{Lop09} the authors extend the notion of cyclicity   in various directions and
generalized in several ways raising the concepts of polycyclic
codes and sequential codes.  Polycyclic codes (formerly known as
pseudo-cyclic codes) are the generalization of cyclic and
constacyclic codes and were studied in \cite[p.241]{PW72} for the
first time. In \cite{HLP08} X.-dong Hou et al.  introduced
sequential codes over finite fields showing that they are
punctured cyclic codes and established a strong connection between
recurrence sequences and sequential codes.  L\'opez-Permouth, et
al.\cite{Lop09} studied the dual of a polycyclic code and showed
\cite[Theorem 3.5]{Lop09} that a constacyclic code is a polycyclic
code whose dual is again polycyclic. Still polycyclic codes never
enjoyed the same popularity that constacyclic codes have because
the dual of some polycyclic code is not polycyclic. In
\cite{Ala16}, an alternate duality is introduced in order to adress
this issue.


In this paper we extend some
results in \cite{Bla17} to polycyclic codes over finite chain rings
using generator polynomials. Galois variance of constacyclic codes over finite fields was
studied, the concept of Galois-disjoint of linear codes was introduced and a new class of codes
so-called polynomial codes were presented. Notice that those
polynomial codes are indeed polycyclic codes over finite fields.

Throughout this work, $S$ is a commutative ring with
identity, and $\texttt{Aut}(S)$ is the group of ring automorphisms
of $S.$ The action of $\texttt{Aut}(S)$ on $S^n$ leads to the
concept of Galois-invariance for any $S$-linear code of length $n$
(see \cite{MNR13}).   In \cite{Lop13} the Hamming weight of some
polycyclic codes over Galois rings was established. Also in
\cite{Fot18} some class of constacyclic codes over finite chain
rings were characterized via cyclotomic cosets as contraction of
some cyclic codes. In  \cite{MNR17}, the authors studied the
algebraic structure of a class of linear codes over finite chain
rings containing the polycyclic codes.

 The structure of the paper is as follows. After some preliminaries in
Section~\ref{sec:2} we use the concept of strong Gr\"obner
bases for polycyclic codes in order to deal with the algebraic
structure of polycyclic codes and their annihilator duals in
Section~\ref{sec:3}. In Section\,\ref{sec:4} we characterize those
polycyclic codes over finite chain rings that are completely
Galois-disjoint and we deduce  an analogue Delsarte theorem relating the
trace map of any polycyclic code and its annihilator dual.

\begin{figure}[htb!]
\centering

\includegraphics[width=0.70\textwidth]{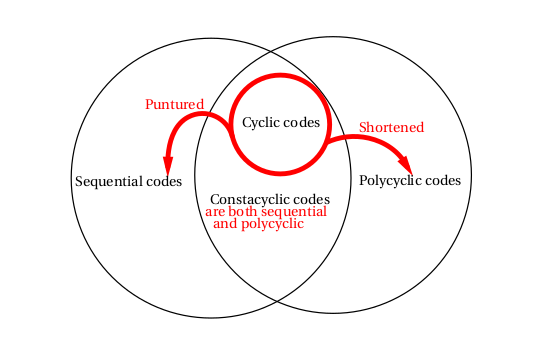}
\end{figure}

\section{Preliminaries}\label{sec:2}

 We use the following notation $p$ be a prime, $a,e,m$ be nonnegative
integers, and $\mathbb{Z}_{p^a}$ be the residue ring of integers
modulo $p^a.$  A ring $S$ is a finite chain ring if $S$ is local
and principal. We will denote  the maximal
ideal of $S$ as  $\texttt{J}(S)=\left<\theta\right>$  and $\mathbb{F}_{p^m}$ the residue
field $S/\texttt{J}(S).$  Thus the ideals of $S$ form a chain under inclusion
$\{0_S\}=\texttt{J}(S)^s\subsetneq\texttt{J}(S)^{s-1}\subsetneq\cdots\subsetneq \texttt{J}(S)\subsetneq S$
and $\texttt{J}(S)^t=\theta^tS$ for $0\leq t< s.$ Note that
 any generator of $\texttt{J}(S)$ is a
root of an Eisenstein polynomial over $\mathbb{Z}_{p^a}$ of degree
$e$ where $p^a$ is the characteristic of $S$ (see for example \cite[Theorem 5.1.]{Nec08}).  The ring
epimorphism $\pi:S\rightarrow\mathbb{F}_{p^m}$ naturally extends a
ring epimorphism from $S[x]$ to $\mathbb{F}_{p^m}[X]$ and on the
other hand it naturally induces an $S$-module epimorphism from
$S^n$ to $(\mathbb{F}_{p^m})^n$, we will abuse the notation  and  thus we also denote it by the same
 $\pi$.

  A polynomial $f$ is basic irreducible over the ring $S$ if
$\pi(f)$ is irreducible over $\mathbb{F}_{p^m}.$   The Galois ring of characteristic $p^a$ and cardinality $p^{am},$
denoted $\texttt{GR}(p^a,m),$ is the quotient ring
$\mathbb{Z}_{p^a}[X]/\langle\,f\,\rangle,$ where
$\langle\,f\,\rangle$ is the ideal generated by a monic basic
irreducible polynomial $f$ of degree $m$ over $\mathbb{Z}_{p^a}.$
Note that for all positive integers $r$ and $r',$ the Galois ring
$\texttt{GR}(p^a,r')$ is a subring of the Galois ring
$\texttt{GR}(p^a,r)$ if and only if $r'$ divides $r.$ Examples of
finite chain rings include Galois rings $\texttt{GR}(p^a,m)$ (here
$\theta =p,$ and $a=s$) and, in particular, finite fields (
$\texttt{GR}(p,m)=\mathbb{F}_{p^m}$).

For a given  finite chain ring $S$ there are integers
 $p,a,m,e,s$  such that  $\mathbb{Z}_{p^a}\subseteq
\mathbb{Z}_{p^a}[\theta]\subseteq S$ and
$S/\texttt{J}(S)\cong\mathbb{F}_{p^m},$ where $\theta$  (the
generator of $\texttt{J}(S)$)  is a root of an Eisenstein polynomial
over $\mathbb{Z}_{p^a}$ of degree $e$ satisfying
$\theta^{s-1}\neq\theta^s=0_S.$  From now on we will denote by $S_r$  the Galois extension of
$\mathbb{Z}_{p^a}[\theta]$ of degree $r$  given by
$S_r=\texttt{GR}(p^a,r)[\theta]$.  Note that a finite chain ring $S$
 is the Galois
extension of $S_r$ of degree $d$ if and only if $m=rd.$ The subset
$\Gamma(S_r)\backslash\{0\}$ is the only cyclic subgroup of
$(S_r)^\times$ isomorphic to $(\mathbb{F}_{p^r})\backslash\{0\},$
where $\Gamma(S_r)=\left\{a \in S_r \,:\, a^{p^r} = a\right\}$ is
the Teichm\"uller set of $S_r.$ Moreover the map
$\pi:S_r\rightarrow\mathbb{F}_{q^r}$ restricted to
$\Gamma(S_r)\backslash\{0\}$ is a group isomorphism from
$\Gamma(S_r)\backslash\{0\}$ onto
$(\mathbb{F}_{p^r})\backslash\{0\}.$ From \cite[Proposition
3.3.]{Nec08} each element $a$ in $S_r$ can be uniquely rewritten
as $a=a_0+a_1\theta+\cdots+a_{s-1}\theta^{s-1},$ where
$(a_0,a_1,\cdots,a_{s-1})\in\Gamma(S_r)^s.$
Thus, the finite chain ring $S_r$ has $p^{sr}$ elements and the units
group $ (S_r)^\times$ of $S_r$ is
$$\left\{\sum_{t=0}^{s-1}a_t\theta^t\,:\,(a_0,a_1,\cdots,a_{s-1})\in\Gamma(S_r)\backslash\{0\}\times(\Gamma(S_r))^{s-1}\right\}.$$

For any divisor $r$ of $m$ the mapping $\sigma^r:S\rightarrow S$
defined by
$\sigma^r\left(\sum_{t=0}^{s-1}a_t\theta^t\right)=\sum_{t=0}^{s-1}a_t^{p^r}\theta^t$
where  $(a_0,a_1,\cdots,a_{s-1})\in\Gamma(S)^s$ is a generator
of the group $\texttt{Aut}_{S_r}(S)$ of the $S_r$-algebra
automorphisms of the ring  $S$, i.e.
$\texttt{Aut}_{S_r}(S)=\left\{\sigma^{ir}\,:\,0\leq i<d\right\}.$
The mapping $\sigma$ is called a \emph{Frobenius automorphism} of
$S.$ It induces the automorphism $\texttt{Fr}_p$ of
$\mathbb{F}_{p^m}$ as follows
$\texttt{Fr}_p(\pi(a))=\pi(\sigma(a))$ for all $a\in S.$

An $S$-linear code $C$ of length $n$ is a submodule of the
$S$-module $S^n.$ An $S$-linear code over $S$ is \emph{free} if
it is free as an $S$-module. A matrix $G$ is called a
\emph{generator matrix} for an $S$-linear code $C$ if the rows of
$G$ span $C$ and none of them can be written as an $S$-linear
combination of the other rows of $G.$ A matrix $G$ is in
\emph{standard form} if there exists an $s$-tuple
$(k_{0},k_{1},\cdots,k_{s-1})$  in $\mathbb{N}^s$ (where $s$ is
the nilpotency index of $S$) such that
\begin{align}\label{stG}G=\left(%
\begin{array}{cccccc}
  I_{k_{0}}  & G_{0,1}        & G_{0,2}        &\cdots    & G_{0,s-1}                 & G_{0,s}                \\
  0        &  \theta I_{k_{1}} &   \theta G_{1,2} & \cdots   &  \theta G_{1,s-1}          & \theta  G_{1,s}         \\
  \cdots   & \cdots         & \cdots         & \cdots   &  \cdots                   & \cdots               \\
  0        &      0         &           0    & \cdots        & \theta^{s-1} I_{k_{s-1}}  &  \theta^{s-1} G_{s-1,s}
\end{array}%
\right)U,
\end{align}
where $I_{k_{t}}$ is an identity matrix of order $k_t$(for $0\leq
t<s$) and $U$ is a
suitable permutation matrix. The $s$-tuple $(k_{0},
k_{1},\cdots,k_{s-1})$ is called \emph{type} of $G.$

From \cite[Proposition 3.2, Theorem 3.5]{NS00}, each $S$-linear
code $C$ admits a generator matrix in standard form, Moreover all  the
generator matrices for  a given code $C$ in standard form have the same type, thus  the type and the rank are numerical invariants for
$S$-linear codes. Henceforth the \emph{type} of any $S$-linear
code is the type of one of its generator matrices in standard
form. The \emph{rank} of an $S$-linear code $C,$ denoted
$\texttt{rank}_S(C),$ is the number of rows of a generator matrix
in standard form. . Notice that
if $C$ is an $S$-linear code of type $(k_{0},
k_{1},\cdots,k_{s-1})$ then $\texttt{rank}_S(C)=k_{0}+
k_{1}+\cdots+k_{s-1},$  also an $S$-linear code $C$ of
type $(k_{0},k_{1},\cdots,k_{s-1})$ is free if and only if
$\texttt{rank}_S(C)=k_0$ and $k_1=k_2=\cdots=k_{s-1}=0.$

 Consider the $n\times n$-matrices
\begin{align}\label{ab}
\begin{array}{ccc}
 \mathrm{D}_{\textbf{a}}= \left(\begin{array}{c|ccc}
    0       &       &                     &     \\
    \vdots  &       & \mathrm{I}_{n-1}    &   \\
    0       &       &                     &   \\
     \hline
    a_0     & a_1   &  \cdots             & a_{n-1}
\end{array}\right) & \text{ and } &
\mathrm{E}_{\textbf{b}}=
  \left(\begin{array}{ccc|c}
    b_0 & \cdots                &  b_{n-2}  &  b_{n-1}   \\
    \hline
        &                       &           &  0 \\
        & \mathrm{I}_{n-1}      &           & \vdots   \\
        &                       &           & 0
\end{array}\right)
\end{array}
\end{align} with associate vectors ${\textbf{a}}=(a_0,a_1,\cdots,a_{n-1})$ and ${\textbf{b}}=(b_0,b_1,\cdots,b_{n-1})
$ in $S^n$. An $S$-linear code of length $n$ is \emph{right
polycyclic}  \cite{Ala16,Lop09} with associate vector
${\textbf{a}}$ (or simply, right ${\textbf{a}}$-\emph{cyclic}) if
it is invariant by right multiplication by the matrix
$\mathrm{D}_{\textbf{a}}$. Also in \cite{Ala16,Lop09} they define
a code as left ${\textbf{b}}$-\emph{cyclic}  if it is  invariant
by right multiplication by the matrix $\mathrm{E}_{\textbf{b}}$.
It is straight forward (adapting the proof in \cite[Theorem
4.]{Ala16} to finite chain rings) to see that any
right ${\textbf{a}}$-cyclic code where the determinant
$\texttt{det}\left(\mathrm{D}_{\textbf{a}}\right)=a_0\in S^\times$
is also left ${\textbf{b}}$-cyclic code  with
$b_j=-a_{j+1}a_0^{-1}$ for $j<n-1$ and $b_{n-1}=a_0^{-1}.$
Therefore from now on  we will assume all right
${\textbf{a}}$-cyclic codes with ${\textbf{a}}\in S^\times\times
S^{n-1}$ and we will simply write ${\textbf{a}}$-cyclic codes.
Note that the $a_0$-constacyclic codes are just
${\textbf{a}}$-\emph{cyclic} codes with
${\textbf{a}}:=(a_0,0,\cdots,0).$

Throughout this paper the
$S$-module monomorphism \begin{align}\label{Psi}
\begin{array}{cccc}
  \Psi: &  S^n & \hookrightarrow & S[X] \\
    & (c_0,c_1,\cdots,c_{n-1}) & \mapsto &
    c_0+c_1X+\cdots+c_{n-1}X^{n-1}.
\end{array}
\end{align} will identify the vectors in $S^n$ with the polynomials
of $S[X]$ with degree at most $n.$  Note that $C$ is an
${\textbf{a}}$-cyclic code if and only if $\Psi(C)$ is an ideal of the quotien ring 
$S[X]/\langle\,
X^n-\Psi(\,{\textbf{a}}\,)\,\rangle$.  We define the
${\textbf{a}}$-period of an  ${\textbf{a}}$-cyclic code over $S$
as the positive integer $\ell_{{\textbf{a}}}$ such that
\begin{align}\ell_{{\textbf{a}}}:=\texttt{min}\left\{i\in\mathbb{N}\backslash\{0\}\;:\;X^n-\pi(\Psi(\,{\textbf{a}}\,))\text{
divides } X^{i}-1\right\}.\end{align}

Note that   the quotient ring $S[X]/\langle\,
X^n-\Psi(\,{\textbf{a}}\,)\,\rangle$ is a principal ideal ring if
either $S$ is a field or $X^n-\pi(\Psi(\,{\textbf{a}}\,))$ is
square free \cite[Theorem
5.2.]{Lop13}.
From now on we will consider   ${\textbf{a}}$-cyclic code over $S$   whose period
is
relatively prime to $p$
and hence
 $X^n-\pi(\Psi(\,{\textbf{a}}\,))$
is square free.
From \cite[Theorem 2.7]{NS00}, if $g\in S[X]$ is monic and
$\pi(g)$ is square-free, then $g$ factors uniquely into monic,
coprime basic irreducibles. For any $f \in \mathbb{F}_{p^m}[X]$
dividing $X^{\ell_{\textbf{a}}} - 1$ such that $p\nmid
\ell_{{\textbf{a}}},$ \cite[Theorem XIII.4]{McD74} implies the
existence   of a unique polynomial $g\in S[X]$ such that
$\pi(g)=f$ and $g$ divides $X^{\ell_{\textbf{a}}} -1$ since
$X^{\ell_{\textbf{a}}} - 1$ is square-free in
$\mathbb{F}_{p^m}[X].$ The polynomial $g$ will be called the
\emph{Hensel lift} of $f.$

\section{Polycyclic codes and their dual codes}\label{sec:3}

In the field case the dual of polycyclic codes  were studied in
\cite{HLP08} via their parity-check matrix. In this section we
generalize structural theorems of polycyclic codes over a finite
field to a finite chain ring $S$.

\subsection{Free polycyclic codes}

Let $\omega\in\{1,2,\cdots,n\}$ and $g=g_0+g_1X +\cdots+g_{k}X^k\in S[X]$
be a polynomial of degree $\texttt{deg}(g):=k$ over $S$ with
$k<n,$ we consider the $\omega\times n$-matrix
\begin{align}
\mathrm{M}_\omega(g):=\left(\begin{array}{ccccccc|ccc}
  g_0      & g_1    & \cdots & g_{k}      & 0      & \cdots & 0       &0&\cdots&0\\
  0        & g_0    & g_1     & \cdots & g_{k}     & 0      & \cdots  &0&\cdots&0\\
  \vdots   & \ddots & \ddots  &  \ddots       &         & \ddots &         &\vdots & &\vdots\\
  0        & \cdots & 0       & g_0     & g_1    & \cdots & g_{k}     &0&\cdots&0
\end{array}%
\right),
\end{align} in which the entries
of the first row correspond to the coefficients of $g,$ and  its
$i^{\text{th}}$ row is the cyclic shift of the $(i-1)^{\text{th}}$
row, for all $i\in\{2,3,\cdots,\omega\}.$ It is easy to see that
if $g_0\in S^\times,$ then the rank of $\mathrm{M}_\omega(g)$ is
$\omega.$

\begin{Definition} A free $S$-linear code $C$ of length $n$ and rank $n-k$ is
\emph{free polycyclic} if there is a monic polynomial $g$ such
that $\mathrm{M}_{n-k}(g)$ is a generator matrix for $C.$
\end{Definition}

Note that for any free polycyclic code $C$ over $S$ there is a unique monic polynomial $g$ over $S$ such that $\mathrm{M}_{n-k}(g)$ is
a generator matrix for $C.$ This polynomial $g$ is called the
\emph{generator polynomial} of $C$ and the free polycyclic code
over $S$ with generator polynomial $f$ of length $n$ is denoted
$\mathcal{P}(S;n;g).$  We will  assume
$\mathcal{P}(S;n;g)=\{\,\textbf{0}\,\},$ if $\texttt{deg}(g)\geq
n$.

\begin{Remark}
Note that
$\pi(\mathcal{P}(S;n;g))=\mathcal{P}(\mathbb{F}_{p^m};n;\pi(g))$, thus
it follows as a generalization of \cite[Theorems 2.3 and
2.4]{Lop09} to finite chain rings  the following
characterization: any free polycyclic code over $S$ with generator
matrix $g$ satisfies $g_0\in S^\times.$
\end{Remark}

From now on, $g=g(X)$ will denote a
monic polynomial in $S[X]$ with $g_0\in S^\times$ and $\textbf{a}$
in $(S^\times)\times(S)^{n-1}$ where $\Psi(\textbf{a})$ is the
{remainder of the Euclidian division} of $X^n$ by $g.$

\begin{Proposition}\label{s3p1}  The following assertions hold.
\begin{enumerate}
    \item $\mathcal{P}(S;n;g)=\biggl\{\textbf{c}\in S^n\;:\;g \text{ divides } \Psi(\,\textbf{c}\,)\biggr\}.$
    \item $\mathcal{P}(S;n;g)$ is  ${\textbf{a}}$-cyclic.
\end{enumerate}
\end{Proposition}

\begin{proof} Let $\textbf{c}=(c_0,c_1,\cdots,c_{n-1})\in \mathcal{P}(S;n;g)=C.$

\begin{enumerate}
  \item  From the definition of $\mathcal{P}(S;n;g),$ it follows that
$M_{n-k}(g)$ is a generator matrix for $C$ if and only if there
is $(m_0,m_1,\cdots,m_{n-k-1})\in S^{n-k}$ such that
$\textbf{c}=(m_0,m_1,\cdots,m_{n-k-1})M_{n-k}(g).$ Note that
$\Psi(\textbf{c})=\Psi(m_0,m_1,\cdots,m_{n-k-1},0,\cdots,0)g.$
Therefore $\textbf{c}\in C,$ if and only if $g$ divides
$\Psi(\textbf{c}).$
\item We have that
$\textbf{c}\mathrm{D}_{\textbf{a}}=(0,c_0,c_1,\cdots,c_{n-2})+c_{n-1}(a_0,a_1,\cdots,a_{n-1}).$
It follows that
$\Psi(\textbf{c}\mathrm{D}_{\textbf{a}})=X\Psi(\textbf{c})-gh,$
since $X^n-\Psi(\textbf{a})=gh$ for some $h\in S[X].$ Hence
$\Psi(\textbf{c}\mathrm{D}_{\textbf{a}})\in C$ because $g$ divides
$\Psi(\textbf{c})$ and therefore $\mathcal{P}(S;n;g)$ is
${\textbf{a}}$-cyclic.
\end{enumerate}

\end{proof}
Let us  follow the notation below.
\begin{itemize}
    \item $\mu(g_1,g_2,\cdots,g_u)$ the Hensel lift of $\texttt{lcm}(\pi(g_1),\pi(g_2),\cdots,\pi(g_u)),$  to $S[X].$
    \item $\delta(g_1,g_2,\cdots,g_u)$ the Hensel lift of $\texttt{gcd}(\pi(g_1),\pi(g_2),\cdots,\pi(g_u)),$  to $S[X].$
\end{itemize}

\begin{Lemma}\label{lem3} Let $\mathcal{P}_i=\mathcal{P}(S;n;g_i)$ where $g_i$ is a  monic polynomial over $S$ such that the period of $\pi(g_i)$ is coprime to
$p$ and $\texttt{deg}(g_i)<n,$ for $1\leq i\leq u.$ The following
assertions hold.
\begin{enumerate}
    \item  $g_i$ divides $g_j$ if and only if $\mathcal{P}_i\supseteq\mathcal{P}_j,$ for all $1\leq i,j\leq u;$
    \item  $\bigcap\limits_{i=1}^u\mathcal{P}_i=\mathcal{P}(S;n;\mu(g_1,\cdots,g_u));$
    \item  $\sum\limits_{i=1}^u\mathcal{P}_i\subseteq\mathcal{P}(S;n;\delta(g_1,\cdots,g_u));$
    \item If $\texttt{deg}(\mu(g_1,\cdots,g_u))\leq n,$ and $\sum\limits_{i=1}^u\mathcal{P}_i$ is free, then $\sum\limits_{i=1}^u\mathcal{P}_i=\mathcal{P}(S;n;\delta(g_1,\cdots,g_u)).$
\end{enumerate}
\end{Lemma}

\begin{proof}

From Proposition~\ref{s3p1} (1) is straight forward. In order to
prove (2) we note that $g_i$ divides $\mu(g_1,\cdots,g_u)$ for
each $i,$  thus from (1) it follows that
$\mathcal{P}_i\supseteq\mathcal{P}(S;n;\mu(g_1,\cdots,g_u))$ for
each $i$ and hence we have the inclusion
$\bigcap\limits_{i=1}^u\mathcal{P}_i\supseteq\mathcal{P}(S;n;\mu(g_1,\cdots,g_u)).$
Conversely, if $\textbf{c}\in
\bigcap\limits_{i=1}^u\mathcal{P}_i$ then by
Proposition~\ref{s3p1}  $\Psi(\textbf{c})$ is a common multiple
of $g_1,\cdots,g_u.$

From \cite[Theorem XIII.6]{McD74}, there are a polynomial over $S$
and $\textbf{c}^*\in S^n$ such that
$\Psi(\textbf{c})=\alpha(X)\Psi(\textbf{c}^*)$
$\Psi(\textbf{c}^*)$ is monic, and  $\pi(\alpha(X))=1.$ Therefore
$\Psi(\textbf{c}^*)$ is also a common multiple of $g_1,\cdots,g_u$
and $\mu(g_1,\cdots,g_u)$ divides $\Psi(\textbf{c}^*).$ Since
$\Psi(\textbf{c})=\alpha(X)\Psi(\textbf{c}^*)$ it follows  that
$\bigcap\limits_{i=1}^u\mathcal{P}_i\subseteq\mathcal{P}(S;n;\mu(g_1,\cdots,g_u))$
proving (3).

Since $\delta(g_1,\cdots,g_u)$ divides $g_i$ for each $i$ by
Proposition~\ref{s3p1}
$\mathcal{P}_i\subseteq\mathcal{P}(S;n;\delta(g_1,\cdots,g_u))$
for each $,$ and since $\mathcal{P}(S;n;\delta(g_1,\cdots,g_u))$
is a submodule of $S^n$ we have
$\sum\limits_{i=1}^u\mathcal{P}_i\subseteq\mathcal{P}(S;n;\delta(g_1,\cdots,g_u)).$
Thus to prove (4) it is enough to show that
$\texttt{rank}_S\left(\sum\limits_{i=1}^u\mathcal{P}_i\right)=\texttt{rank}_S(\mathcal{P}(S;n;\delta(g_1,\cdots,g_u))),$
if $\texttt{deg}(\mu(g_1,\cdots,g_u)\leq n$ and
$\sum\limits_{i=1}^u\mathcal{P}_i$ is free. Since
$\texttt{deg}(\texttt{gcd}(\pi(g_1),\cdots,\pi(g_u)))=\texttt{deg}(\mu(g_1,\cdots,g_u)\leq
n$ by \cite[Lemma 5]{Bla17} it follows that
$$\sum\limits_{i=1}^u\pi(\mathcal{P}_i)=\mathcal{P}(\mathbb{F}_{p^m};n;\texttt{gcd}(\pi(g_1),\cdots,\pi(g_u))),$$
and
$$\texttt{rank}_S\left(\sum\limits_{i=1}^u\mathcal{P}_i\right)=\texttt{dim}_{\mathbb{F}_{p^m}}\left(\pi\left(\sum\limits_{i=1}^u\mathcal{P}_i\right)\right).$$
Hence
$\texttt{rank}_S\left(\mathcal{P}(S;n;\delta(g_1,\cdots,g_u))\right)=\texttt{dim}_{\mathbb{F}_{p^m}}\left(\mathcal{P}(\mathbb{F}_{p^m};n;\texttt{gcd}(\pi(g_1),\cdots,\pi(g_u))\right)$
and this completes the proof.
\end{proof}

\subsection{Non-free polycyclic codes}

Let $C$ be a non-free non-zero  ${\textbf{a}}$-cyclic code  of length $n$ over
$S,$ and let $g_{0}$ be the polynomial $X^n-\Psi(\,\textbf{a}).$

\begin{Definition}
A subset $\left\{\theta^{\lambda_1}g_1\,;\,\theta^{\lambda_2}
g_2,\cdots,\theta^{\lambda_u}g_u\right\}$ of $S[X]$ (for some
$1\leq u< s $) is a \emph{strong Gröbner basis} for $C$ if
$\Psi(C)=\left\langle\theta^{\lambda_1}g_1\,;\,\theta^{\lambda_2}
g_2,\cdots,\theta^{\lambda_u}g_u\right\rangle$ and
\begin{enumerate}
    \item $0\leq\lambda_1<\lambda_2<\cdots<\lambda_u<s;$
    \item for $1\leq i\leq u,$ $g_i$ is monic;
    \item $n>\texttt{deg}(g_1)>\texttt{deg}(g_2)>\cdots>\texttt{deg}(g_u)> 0;$
    \item for $0\leq i\leq u,$  $\theta^{\lambda_{i+1}}g_i\in\left\langle\theta^{\lambda_{i+1}}g_{i+1},\theta^{\lambda_{i+2}}g_{i+2},\cdots,\theta^{\lambda_u}g_u\right\rangle$ in $S[X].$
\end{enumerate}
\end{Definition}

 \begin{Lemma}(\cite[Theorem 4.2]{Sal03}.) Every non-free, non-zero polycyclic code over
$S$ admits a strong Gröbner basis.
\end{Lemma}

\begin{Remark}
Condition (4) in the definition  implies that $\pi(g_u) |\pi(g_{u-1})| \cdots
|\pi(g_1)|\pi(g_0).$ Note that for polycyclic codes that condition cannot be
improved in general: there are polycyclic codes for which no
strong Gröbner basis
$\left\{\theta^{\lambda_1}g_1\,,\,\theta^{\lambda_2}
g_2,\cdots,\theta^{\lambda_u}g_u\right\}$ has the property $g_u
|g_{u-1}| \cdots |g_1|g_0$ (see \cite[Example 3.3]{Sal03}).
\end{Remark}

A strong Gröbner basis
$\left\{\theta^{\lambda_1}g_1\,,\,\theta^{\lambda_2}
g_2,\cdots,\theta^{\lambda_u}g_u\right\}$ for the code $C$ as
described above is not necessarily unique. However, if $k_i$ is
the degree of $g_i,$ then $u,$ $(k_1,k_2,\cdots,k_u)$ and
$(\lambda_1,\lambda_2,\cdots,\lambda_u)$ are uniquely determined
by the algebraic structure of the code $C.$ Therefore  the type of
$C$ is $(k_0,k_1,\cdots,k_{s-1})$ where $k_t=k_{i-1}-k_i$ if
$\lambda_i=t$ for some $i,$ $k_t=0$ otherwise, and
$|C|=\prod\limits_{i=1}^{u}p^{m(s-\lambda_i)(k_{i-1}-k_i)}.$

\begin{Remark}\label{rem3} Let $C$ be a non-free, non-zero  ${\textbf{a}}$-cyclic code over $S$ such that
$\left\{\theta^{\lambda_1}g_1\,,\,\theta^{\lambda_2}
g_2,\cdots,\theta^{\lambda_u}g_u\right\}$ is one of its strong
Gröbner bases. Let us  denote by  $k_i=\texttt{deg}(g_i)$ for
$i\in\{0,1,\cdots,u\}.$ Then a generator matrix for $C$ is
\begin{align}\left(
\begin{array}{c}
  \theta^{\lambda_0}\mathrm{M}_{k_0-k_1}(g_1) \\
  \theta^{\lambda_1}\mathrm{M}_{k_1-k_2}(g_2)\\
  \vdots\\
\theta^{\lambda_u}\mathrm{M}_{k_{u}-k_{u-1}}(g_u)
\end{array}%
\right).\end{align}
 \end{Remark}

\begin{Example}\label{ex1} Consider the ring  $S=\mathbb{Z}_4[\alpha]$ where
$\alpha^2=3\alpha+3$ and  its residue field
$\mathbb{F}_4=\mathbb{F}_2(\beta)$ where  $\beta^2=\beta+1.$ The
factorization of $X^5-\beta$  is the product of irreducible
polynomials in $\mathbb{F}_4[X]$ given by
 $$X^5-\beta=(X-\beta^2)(X^2+ X+\beta)(X^2+\beta
X+\beta).$$  Let us take $g_2=X^2+ X+\alpha,$
$g_1=X^4+\alpha^2X^3+\alpha X^2+X+\alpha^2$ and
$g_0=(X-\alpha^2)g_1+2(X+3)g_2=X^5+2X^3+2\alpha^2 X+\alpha$ in
$S[X].$  Since $2g_1=2g_2(X^2+\alpha X+\alpha)$ if follows that
the code $C$ with Gröbner basis $\left\{g_1\,,\,2 g_2\right\}$ is
an $\textbf{a}$-cyclic code of length $5$ whose
$\textbf{a}:=(\alpha,2\alpha^2,0,2)$ and
$$
\left(%
\begin{array}{ccccc}
  \alpha^2 & 1 & \alpha & \alpha^2 & 1 \\
  2\alpha & 2 & 2 & 0 & 0 \\
  0 & 2\alpha & 2 & 2 & 0
\end{array}%
\right)
$$ is  one of its generator matrices.
\end{Example}

\subsection{Annihilator dual of a linear code}

The \emph{Euclidian product} on $S^n$ is  $\langle
\textbf{x}\,;\,\textbf{y}\rangle=
\textbf{x}\cdot\textbf{y}^{\texttt{tr}},$ where $\textbf{x},\textbf{y}\in S^n.$  and
$\textbf{y}^{\texttt{tr}},$  the transpose of $\textbf{y}\in S^n$. The Euclidean orthogonal code $C^{\perp}$ of
an $S$-linear code $C$ of length $n,$ is defined  as
$C^{\perp}:=\left\{\textbf{y}\in S^n\;:\; \langle
\textbf{y}\,;\,\textbf{c}\rangle=0_{S}\;\text{ for all }
\textbf{c}\in C \right\}.$   $C^{\perp}$ is an $S$-linear
code of length $n$ and from \cite[Theorem 3.10(iii)]{NS00},
$\left(C^{\perp}\right)^{\perp}=C$ and
$|C^{\perp}|\times|C|=|S|^n.$

\begin{Definition} An $S$-linear code $C$
is ${\textbf{a}}$-\emph{sequential} if
$C\cdot\mathrm{D}_{\textbf{a}}^{\texttt{tr}} =C$ where
$C\cdot\mathrm{D}_{\textbf{a}}^{\texttt{tr}}=\left\{\textbf{c}\mathrm{D}_{\textbf{a}}^{\texttt{tr}}\,:\,\textbf{c}\in
C\right\}$ and $\mathrm{D}_{\textbf{a}}^{\texttt{tr}}$ is the
transpose of $\mathrm{D}_{\textbf{a}}.$
\end{Definition}

\begin{Proposition} Any $S$-linear code $C$ of length $n$ is
 ${\textbf{a}}$-cyclic if and only if  $C^{\perp}$ is
${\textbf{a}}$-sequential.
\end{Proposition}

\begin{proof} Assume that $C\cdot\mathrm{D}_{\textbf{a}}^{\texttt{tr}}=C.$ If $\textbf{y}\in C^{\perp}$ then
$\left\langle\textbf{y}\,;\,\left(\textbf{c}\mathrm{D}_{\textbf{a}}\right)^{\texttt{tr}}\right\rangle
=\left(\textbf{y}\mathrm{D}_{{\textbf{a}}}^{\texttt{tr}}\right)\cdot\textbf{c}^{\texttt{tr}}=0_S$
for all $\textbf{c}\in C$ and therefore $\textbf{y}\mathrm{D}_{{\textbf{a}}}^{\texttt{tr}}\in C^{\perp}$
 and
$C^{\perp}\cdot\mathrm{D}_{{\textbf{a}}}^{\texttt{tr}}\subseteq
C^{\perp}.$ Since $\mathrm{D}_{{\textbf{a}}}$ is invertible, we
have $C^{\perp}\cdot\mathrm{D}_{{\textbf{a}}}^{\texttt{tr}}=
C^{\perp}.$ Conversely, we have
$\left(C^{\perp}\cdot\mathrm{D}_{{\textbf{a}}}^{\texttt{tr}}\right)^{\perp}=
C,$ since $\left(C^{\perp}\right)^{\perp}=C.$ Thus it follows
$\left(C^{\perp}\cdot\mathrm{D}_{{\textbf{a}}}^{\texttt{tr}}\right)^{\perp}=\left(C^{\perp}\right)^{\perp}\cdot\mathrm{D}_{{\textbf{a}}}$
 and $C\cdot\mathrm{D}_{\textbf{a}}=C.$
\end{proof}

\begin{Remark} According to \cite[Theorem 3.5.]{Lop09}, an ${\textbf{a}}$-sequential code of length $n$
is identified to an ideal in
$S[X]/\langle\,X^n-\Psi(\textbf{a})\,\rangle$ if and only if it is
constacyclic.
\end{Remark}

For a polynomial $\rho\in S[X]$ we will denote by
$\{\rho\}_{\textbf{a}}$ the remainder of the Euclidian division of
$\rho$ by $X^n-\Psi(\textbf{a}).$ We introduce the inner-product
on $S^n$ defined by: \begin{align}\label{in1}\langle
\textbf{u}\,;\,\textbf{v}\rangle_{\textbf{a}}=
\left(\left\{\Psi(\textbf{u})\Psi(\textbf{v})\right\}_{\textbf{a}}\right)(0_S).\end{align}

\begin{Lemma} The inner-product (\ref{in1})
is a nondegenerate symmetric bilinear form.
\end{Lemma}

\begin{proof} It is clear that  it is a symmetry bilinear form. Since $\textbf{a}\in
(S^\times)\times(S^{n-1})$ it follows that $1$ is the remainder
of the Euclidian division of $Xf$ by $X^n-\Psi(\textbf{a})$ for
some $f\in S[X].$ Thus
$\left(\left\{\Psi(\textbf{c})X^i\right\}_{\textbf{a}}\right)(0_S)=0_S$
for all integer $i$ if and only if $\textbf{c}=\textbf{0}.$
\end{proof}

The result in  \cite[Proposition 1]{Ala16} allows to write the
\emph{annihilator dual} $C^{\circ}$  of a $S$-linear code $C$
 as follows: $$C^{\circ}:=\left\{\textbf{y}\in
S^n\,:\,\langle\textbf{y}\,;\,\textbf{c}\rangle_{\textbf{a}}=0_S\text{
for all } \textbf{c}\in C\,\right\}.$$

\begin{Lemma}\label{circ0} Let $\mathrm{A}=(\langle
\textbf{e}_i\,;\,\textbf{e}_j\rangle_{\textbf{a}})$ be a matrix
with respect to a basis $\{\textbf{e}_i\;:\;1\leq i\leq n\}$  of the
$S$-module $S^n$ and $C$ an $S$-linear code.  Then
$C^{\circ}=(C\cdot\mathrm{A})^\perp$ where
$C\cdot\mathrm{A}:=\{\textbf{c}\mathrm{A}\;:\;\textbf{c}\in C\}.$
\end{Lemma}

\begin{proof} Note that
$\langle\textbf{u}\,;\,\textbf{v}\rangle_{\textbf{a}}=\textbf{u}\mathrm{A}\textbf{v}^t=\langle
\textbf{u}\,;\,\textbf{v}\mathrm{A}\rangle_{\textbf{a}}$ where
$\textbf{v}^t$ is the transpose of $\textbf{v}.$ Thus
$$ C^{\circ} = \left\{\textbf{y}\in S^n\,:\,\langle\textbf{y}\,;\,\textbf{c}\mathrm{A}\rangle_{\textbf{a}}=0_S\text{ for all } \textbf{c}\in C\,\right\}
= \left\{\textbf{y}\in S^n\,:\,\langle
\textbf{y}\,;\,\textbf{c}\rangle=0_S\text{ for all } \textbf{c}\in
C\cdot\mathrm{A}\,\right\}= (C\cdot\mathrm{A})^\perp.$$
\end{proof}

\begin{Corollary}\label{circ} Let $C$ be an $S$-linear code. The following assertions hold.
\begin{enumerate}
    \item $C^\perp$ and $C^{\circ}$ have the same type;
    \item $\left(C^{\circ}\right)^{\circ}=C.$
\end{enumerate}
\end{Corollary}

\begin{proof}  Let $\{\textbf{e}_i\;:\;1\leq i\leq n\}$ be a basis of the free $S$-module $S^n$ and $\mathrm{A}=(\langle
\textbf{e}_i\,;\,\textbf{e}_j\rangle_{\textbf{a}}).$
\begin{enumerate}
  \item The matrix $\mathrm{A}$ is invertible, since
$\langle
\textbf{u}\,;\,\textbf{v}\rangle_{\textbf{a}}$ is nondegenerate.
Hence the $S$-modules $(C\mathrm{A})^\perp$ and $C^\perp$ are
isomorphic and whence $C^\perp$ and $C^{\circ}$ have the same
type.
\item Using the equality $C^{\circ}=(C\cdot\mathrm{A})^\perp$ and
the fact that  $\mathrm{A}$ is invertible it follows that
\begin{eqnarray*}
  \left(C^{\circ}\right)^{\circ} &=& (C^{\circ}\cdot\mathrm{A})^\perp \\
    &=& (C^{\circ})^\perp(\mathrm{A}^{-1})  \\
    &=& ((C\cdot\mathrm{A})^\perp)^\perp\cdot(\mathrm{A}^{-1})=(CA)A^{-1}=C.
\end{eqnarray*}
\end{enumerate}
\end{proof}

\begin{Proposition} If $C$ be an ${\textbf{a}}$-cyclic code over
$S$ then $C^{\circ}$ is  an ${\textbf{a}}$-cyclic code.
\end{Proposition}

\begin{proof}
$X\Psi(\textbf{y})-\Psi(\textbf{y}\mathrm{D}_{\textbf{a}})\in\langle\,X^n-\Psi(\textbf{a)}\,\rangle$
for all $\textbf{y}\in S^n$ therefore it follows that
$C^{\circ}\cdot\mathrm{D}_{\textbf{a}}=(C\cdot\mathrm{D}_{\textbf{a}})^{\circ}.$
By hypothesis $C\cdot\mathrm{D}_{\textbf{a}}=C$ and hence we have
that $C^{\circ}\cdot\mathrm{D}_{\textbf{a}}=C^{\circ}.$
\end{proof}

\begin{Proposition} Let $C:=\mathcal{P}(S;n;g)$ be a nonzero ${\textbf{a}}$-cyclic code and $h$ a monic polynomial such that
$gh=X^n-\Psi(\textbf{a}).$ Then
$(\mathcal{P}(S;n;g))^{\circ}=\mathcal{P}(S;n;h).$
\end{Proposition}

\begin{proof} Let $\textbf{y}\in S^n$,  there exist $t\in\mathbb{N},$ $\alpha(X)\in(S[X])^\times$ and $\textbf{y}^*\in S^n$
such that $\Psi(\textbf{y})=\theta^t\alpha(X)\Psi(\textbf{y}^*)$
and $\Psi(\textbf{y}^*)$ is monic (see \cite[Theorem
XIII.6]{McD74}). If $\textbf{y}\in(\mathcal{P}(S;n;g))^{\circ},$
then $\Psi(\textbf{y}^*)g\in\langle\,gh\,\rangle$ since
$gh=X^n-\Psi(\textbf{a}).$ $h$ divides
$\Psi(\textbf{y}^*)$ thus
$(\mathcal{P}(S;n;g))^{\circ}\subseteq \mathcal{P}(S;n;h).$
Conversely $\mathcal{P}(S;n;h)\subseteq
(\mathcal{P}(S;n;g))^{\circ}$ since $gh=X^n-\Psi(\textbf{a}).$
\end{proof}

\section{Galois variance of polycyclic codes}\label{sec:4}

In this section the integers $r$ and $d$ are two divisors of $m$ such that $m=dr.$

\subsection{Galois images of polycyclic codes}\label{susec:41}

Let $\sigma$ be a generator of $\texttt{Aut}(S).$ This generator
$\sigma$ naturally induces on the one hand  an automorphism of
$S[X]$ as follows:
$$\sigma(f)=\sum\limits_{i}\sigma(f_i)X^i,$$
for any $f=\sum\limits_{i}f_iX^i\in S[X]$ and the other hand an
automorphism of $S^n$ of the following way:
$$
\sigma(\,\textbf{c}\,)=\left(\sigma(c_0),\sigma(c_1),\cdots,\sigma(c_{n-1})\right),$$
for any $\textbf{c}:=(c_0,c_1,\cdots,c_{n-1})\in S^n$ respectively. Given a positive integer $i,$ the $S$-linear code
$\sigma^i(C):=\left\{ \sigma^i(\,\textbf{c}\,)\;:\;\textbf{c}\in
C\right\}$ is called \emph{image} of $C$ under $\sigma^i.$

\begin{Definition} Let $C$ be an $S$-linear code of length $n.$
\begin{enumerate}
    \item  The code $C$ is  $\langle\sigma^r\rangle$-\emph{disjoint} if $\sigma^{ir}(C)\cap C=\{\textbf{0}\},$ for all $1\leq i <d.$
     \item  The code $C$ is complete $\langle\sigma^r\rangle$-\emph{disjoint} if $S^n=C\oplus\sigma^r(C)\oplus\cdots\oplus\sigma^{r(d-1)}(C).$
\end{enumerate}
\end{Definition}

Note that any complete $\langle\sigma^r\rangle$-disjoint code
over $S$ is $\langle\sigma^r\rangle$-disjoint, on the other hand
if an $S$-linear code $C$ is $\langle\sigma^r\rangle$-disjoint
then all subcodes and $\langle\sigma^i\rangle$-images of $C$ are
also $\langle\sigma^i\rangle$-disjoint and from Remark~\ref{rem3}
we have the following result.

\begin{Proposition} Let $C$ be a non-free, non-zero  ${\textbf{a}}$-cyclic code over $S$ whose
$\left\{\theta^{\lambda_1}g_1\,,\,\theta^{\lambda_2}
g_2,\cdots,\theta^{\lambda_u}g_u\right\}$ is one of its strong
Gröbner bases. Then  $\sigma^i(C)$ is
$\sigma^i({\textbf{a}})$-cyclic, and
$\left\{\theta^{\lambda_i}\sigma^i(g_i)\,:\,1\leq i\leq u\right\}$
is a Gröbner basis for $\sigma^i(C),$ for all $0\leq i<m.$
\end{Proposition}

\begin{Proposition}\label{pro4} The code $C:=\mathcal{P}(S;n;g)$ is $\langle\sigma^r\rangle$-disjoint if and only if
$\texttt{deg}(\mu(g,\sigma^{ir}(g)))\geq n,$  for $1\leq i\leq
d-1.$
\end{Proposition}

\begin{proof} Note that $C$ is $\langle\sigma^r\rangle$-disjoint if and only
if $\sigma^{ir}(C)\cap C=\{\textbf{0}\}$ for all $1\leq i <d.$ By
part 2. of Lemma~\ref{lem3} this is possible if and only if
$\mathcal{P}(S;n;\mu(g,\sigma^{ir}(g)))=\{\textbf{0}\},$ which
happens if and only if $\texttt{deg}(\mu(g,\sigma^{ir}(g)))\geq
n.$
\end{proof}

Note that if
$S^n=C\oplus\sigma^r(C)\oplus\cdots\oplus\sigma^{r(d-1)}(C),$ and
since $S^n$ is a free $S$-module, therefore  $C$ is
a free $S$-module. We have the following result.

\begin{Lemma} Let $C$ be an $S$-linear code over $S$ of length $n.$  If $C$ is complete
$\langle\sigma^r\rangle$-disjoint then $C$ is free.
\end{Lemma}

The following example gives a $\langle\sigma^r\rangle$-disjoint
code which is non-free.

\begin{Example} Let $\sigma$ be a generator of
$\texttt{Aut}(\mathbb{Z}_4[\alpha])$ defined by
$\sigma(\alpha)=\alpha^2.$ Consider the $\textbf{a}$-cyclic code
$C$ over $\mathbb{Z}_4[\alpha]$ with Gröbner basis
$\left\{g_1\,;\,2 g_2\right\},$ given in Example~\ref{ex1}. A
generator matrix of $\sigma(C)$ is
$$
\left(%
\begin{array}{ccccc}
  3\alpha & 3 & 3\alpha^2 & \alpha^2+3 & 1 \\
  2\alpha^2 & 2 & 2 & 0 & 0 \\
  0 & \alpha^2 & 2 & 2 & 0
\end{array}%
\right),
$$ and $\sigma(C)$ is an $\sigma(\textbf{a})$-cyclic code of
length $5$ with
$\textbf{a}=(3\alpha^2,2\alpha^2,2\alpha^2,2\alpha^2,2\alpha^2).$
A generating set of $C\cap\sigma(C)$ is $C\cap(\mathbb{Z}_4)^5$
(see \cite{Fot16}), but $C\cap(\mathbb{Z}_4)^5=\{\textbf{0}\}.$
Hence $C$ is Galois disjoint.
\end{Example}

\begin{Theorem} The code $C=\mathcal{P}(S;n;g)$ is complete $\langle\sigma^r\rangle$-disjoint if and only if
$\texttt{deg}(g)=(d-1)\frac{n}{d}$ and
$\texttt{deg}(\mu(g,\sigma^{ir}(g)))\geq n,$ for $1\leq i\leq
d-1.$
\end{Theorem}

\begin{proof} We note that $C$ is complete $\langle\sigma^r\rangle$-disjoint if and only
if $C$ is $\langle\sigma^r\rangle$-disjoint and
$\texttt{rank}_S(C)=\frac{n}{d}.$ Now
$\texttt{rank}_S(C)=n-\texttt{deg}(g)$ and by
Proposition~\ref{pro4} it follows that $C$ is complete
$\langle\sigma^r\rangle$-disjoint if and only if
$\texttt{deg}(g)=(d-1)\frac{n}{d}$ and
$\texttt{deg}(\mu(g,\sigma^{ir}(g)))\geq n$ for $1\leq i\leq
d-1.$
\end{proof}

\subsection{Trace and restriction of polycyclic codes}\label{subsec:42}

  Let $C$ be an $S$-linear code of length $n.$ The \emph{restriction code} $\texttt{Res}_r(C)$ of $C$ to $S_r$ is
defined to be $\texttt{Res}_r(C):=C\cap (S_r)^n,$ and the
\emph{trace code} $\texttt{Tr}_d(C)$ of $C$ to $S_r$ is defined
as
$$\texttt{Tr}_d(C):=\biggl\{\left(\texttt{Tr}_d(c_0),\texttt{Tr}_d(c_1),\cdots,\texttt{Tr}_d(c_{n-1})\right)\;:\;(c_0,c_1,\cdots,c_{n-1})\in
C\biggr\},$$ where
$\texttt{Tr}_d=\sum\limits_{i=0}^{d-1}\sigma^{ir}.$ The following result describes
 the restriction code and trace code of any free polycyclic
code.

\begin{Theorem} Let $C=\mathcal{P}(S;n;g).$ Set $\mu_d(g):=\mu(g,\sigma^{r}(g),\cdots,\sigma^{r(d-1)}(g))$ and
$\delta_d(g)=\delta(g,\sigma^{r}(g),\cdots,\sigma^{r(d-1)}(g))$.
\begin{enumerate}
    \item $\texttt{Res}_r(C)=\mathcal{P}(S_r;n;\mu_d(g)).$
    \item
    $\texttt{Tr}_d(C)\subseteq\mathcal{P}(S_r;n;\delta_d(g)),$
    with equality if $\sum\limits_{i=0}^{d-d}\sigma^{ir}(C)$ is free and $\texttt{deg}(\mu_d(g))\leq n.$
\end{enumerate}
\end{Theorem}

\begin{proof}$\quad $

\begin{enumerate}
  \item By the first part of Lemma~\ref{lem3} we note that $\mathcal{P}(S_r;n;\mu_d(g))\subseteq C$ and is in $(S_r)^n,$ thus
$\mathcal{P}(S_r;n;\mu_d(g))\subseteq\texttt{Res}_r(C).$
Conversely, if $\textbf{a}\in\texttt{Res}_r(C),$ then
$\textbf{a}=\sigma^{ir}(\textbf{a})$ for all $0\leq i\leq d-1,$ therefore
$\textbf{a}\in\bigcap\limits_{i=0}^{d-1}\sigma^{ir}(C)$ and from
the second part  of Lemma~\ref{lem3} it follows that
$\textbf{a}\in\mathcal{P}(S_r;n;\mu_d(g)).$ Hence
$\texttt{Res}_r(C)=\mathcal{P}(S_r;n;\mu_d(g)).$

\item From the third part of
Lemma~\ref{lem3} it follows that
$\sum\limits_{i=0}^{d-d}\sigma^{ir}(C)\subseteq\mathcal{P}(S;n;\delta_d(g)).$
Now
$\texttt{Tr}_d\left(\sum\limits_{i=0}^{d-d}\sigma^{ir}(C)\right)=\texttt{Tr}_d(C)$
and
$\texttt{Tr}_d\left(\mathcal{P}(S;n;\delta_d(g))\right)=\mathcal{P}(S_r;n;\delta_d(g))$
since $\sigma^{r}(\delta_d(g))=\delta_d(g).$ Thus
$\texttt{Tr}_d(C)\subseteq\mathcal{P}(S_r;n;\delta_d(g))$,
moreover, if $\sum\limits_{i=0}^{d-d}\sigma^{ir}(C)$ is free and
$\texttt{deg}(\mu_d(g))\leq n,$ then by part 4 of
Lemma~\ref{lem3}
$\sum\limits_{i=0}^{d-d}\sigma^{ir}(C)=\mathcal{P}(S;n;\delta_d(g)).$
Whence it follows the equality
$\texttt{Tr}_d(C)=\mathcal{P}(S_r;n;\delta_d(g)).$
\end{enumerate}

\end{proof}

Finally the following result gives an analogue Delsarte's theorem relating the
trace map of a $S$-linear code and, in this case, its annihilator dual.

\begin{Proposition}\label{Delsarte} For any $S$-linear code $C$
the following holds $$\texttt{Tr}_d(C^{\circ})=(\texttt{Res}_r(C))^{\circ}.$$
\end{Proposition}

\begin{proof} Let $\textbf{x}\in\texttt{Tr}_d(C^{\circ})$,
 there is $\textbf{y}\in C^{\circ}$ such that
$\textbf{x}=\texttt{Tr}_d(\textbf{y}).$
 For all
$\textbf{c}\in\texttt{Res}_r(C),$
\begin{eqnarray*}
  \left\langle\,\textbf{x},\textbf{c}\,\right\rangle_{\textbf{a}} &=& \left\langle\,\texttt{Tr}_d(\textbf{y}),\textbf{y}\,\right\rangle_{\textbf{a}}, \\
    &=&
    \texttt{Tr}_d\left(\left\langle\,\textbf{y},\textbf{z}\,\right\rangle_{\textbf{a}}\right), \\
    &=&\texttt{Tr}_d(0_ S),\\
    &=&0_S.
\end{eqnarray*}
Thus $\texttt{Tr}_d(
C^{\circ})\subseteq(\texttt{Res}_r(C))^{\circ}.$ On the other hand
if $\textbf{x}\in\left(\texttt{Tr}_d(C^{\circ})\right)^{\circ}$
then
$\left\langle\,\textbf{x},\texttt{Tr}_d(\textbf{z})\right\rangle_{\textbf{a}}=0_S$
for all $\textbf{z}$ in $ C^{\circ}.$ Now
$\texttt{Tr}_d(\left\langle\,\textbf{x},\lambda\textbf{z}\,\right\rangle)=\left\langle\,\textbf{x},\texttt{Tr}_d(\lambda\textbf{z})\right\rangle_{\textbf{a}}=0_S$
and
$\left\langle\,\textbf{x},\lambda\textbf{z}\right\rangle_{\textbf{a}}=\lambda(\left\langle\,\textbf{x},\textbf{z}\right\rangle_{\textbf{a}})$
for all $\lambda$ in $ S.$ Since $(x,y)\mapsto\texttt{Tr}_d(x,y)$
is a symmetry nondegenerate bilinear form it follows that
$\left\langle\,\textbf{x},\textbf{z}\right\rangle_{\textbf{a}}=0_S.$
Hence
$\left(\texttt{Tr}_d(C^{\circ})\right)^{\circ}\subseteq\texttt{Res}_r(C),$
and so $(\texttt{Res}_r(C))^{\circ}\subseteq\left(\texttt{Tr}_d(
C^{\circ})^{\circ}\right)^{\circ}=\texttt{Tr}_d( C^{\circ}),$ by
Corollary\,\ref{circ}. Therefore the equality
$\texttt{Tr}_d(C^{\circ})=(\texttt{Res}_r(C))^{\circ}.$
\end{proof}

\section{Conclusion}

A structural theorem for polycyclic codes over $S$ has been
established  using the concept of strong Gröbner bases. We have
explored the action of the Galois group $\texttt{Aut}(S)$ with
generator $\sigma$ on polycyclic codes over $S$ of length $n$
whose period is relatively prime to $p.$ For any divisor $r$ of
$m,$ the complete $\langle\,\sigma^r\,\rangle$-disjoint polycyclic
codes have been characterized. Finally an analogue Delsarte's for
the annihilator dual has been established. Although the trace
codes of free polycyclic codes have been determined, the
construction of a strong Gröbner basis for trace codes of non-free
polycyclic codes will be the object of further study.

\end{document}